\documentclass[USenglish,cleveref,nameinlink,autoref,thm-restate]{lipics-v2021}
\usepackage{graphicx} 
\usepackage{xspace}
\usepackage{amsmath,amssymb,amsthm,amsfonts}
\usepackage{todonotes}
\usepackage{complexity}
\usepackage{hyperref}
\usepackage[nocompress,noadjust]{cite}
\usepackage[inline]{enumitem}

\hideLIPIcs

\makeatletter
\def\renewtheorem#1{%
  \expandafter\let\csname#1\endcsname\relax
  \expandafter\let\csname c@#1\endcsname\relax
  \gdef\renewtheorem@envname{#1}
  \renewtheorem@secpar
}
\def\renewtheorem@secpar{\@ifnextchar[{\renewtheorem@numberedlike}{\renewtheorem@nonumberedlike}}
\def\renewtheorem@numberedlike[#1]#2{\newtheorem{\renewtheorem@envname}[#1]{#2}}
\def\renewtheorem@nonumberedlike#1{  
\def\renewtheorem@caption{#1}
\edef\renewtheorem@nowithin{\noexpand\newtheorem{\renewtheorem@envname}{\renewtheorem@caption}}
\renewtheorem@thirdpar
}
\def\renewtheorem@thirdpar{\@ifnextchar[{\renewtheorem@within}{\renewtheorem@nowithin}}
\def\renewtheorem@within[#1]{\renewtheorem@nowithin[#1]}
\makeatother

\usepackage{mdframed}
\theoremstyle{plain}
\renewtheorem{lemma}[theorem]{Lemma}
\crefname{lemma}{Lemma}{Lemmas}
\Crefname{lemma}{Lemma}{Lemmas}
\crefname{corollary}{Corollary}{Corollaries}
\Crefname{corollary}{Corollary}{Corollaries}

\usepackage{apptools,thmtools}
\let\origrestatable=\restatable
\renewcommand{\restatable}[1][]{%
  \origrestatable[%
    \ifstrempty{#1}%
      {$\star$}%
      {\IfAppendix{\hyperref[#1]{$\star$}}{\hyperref[#1*]{$\star$}}}%
  ]%
}

\renewcommand{\S}{\mathcal{S}}

\usepackage{xcolor}
\definecolor{defblue}{rgb}{0.121,0.47,0.705}
\DeclareTextFontCommand{\emph}{\color{defblue}\em}
\hypersetup{colorlinks=true,
    linkcolor=defblue,
    anchorcolor=defblue,
    citecolor=defblue,
    filecolor=defblue,
    menucolor=defblue,
    urlcolor=defblue,
    bookmarksopen=true,
    bookmarksopenlevel=2,
    bookmarksnumbered=true,
    plainpages=false
    }
\BeforeBeginEnvironment{mdframed}{\nolinenumbers}
\AfterEndEnvironment{mdframed}{\nolinenumbers}
\title{On the Complexity of Extending Storylines}
\author{Alexander~Dobler}{TU Wien, Austria}{adobler@ac.tuwien.ac.at}{}{}
\author{Siddharth~Gupta}{BITS Pilani, K K Birla Goa Campus, India}{siddharthg@goa.bits-pilani.ac.in}{}{}
\author{Philipp~Kindermann}{Trier University, Germany}{kindermann@uni-trier.de}{}{}
\author{Fabrizio~Montecchiani}{University of Perugia, Italy}{fabrizio.montecchiani@unipg.it}{}{}
\author{Martin~Nöllenburg}{TU Wien, Austria}{noellenburg@ac.tuwien.ac.at}{}{}

\authorrunning{A.~Dobler, S.~Gupta, P.~Kindermann, F.~Montecchiani, and M.~Nöllenburg}

\Copyright{Alexander~Dobler, Siddharth~Gupta, Philipp~Kindermann, Fabrizio~Montecchiani, and Martin~Nöllenburg}

\newtheorem{property}{Property}


\acknowledgements{This work started at the Bertinoro Workshop on Graph Drawing BWGD 2025.}

\newcommand{\defquestion}[3]{
	\vspace{-2mm}
	\noindent\begin{mdframed}
		\begin{minipage}{0.96\linewidth}
			\begin{tabular*}{\linewidth}{@{\extracolsep{\fill}}lr} \textsc{#1} & \\ \end{tabular*}
			{\bf{Input:}} #2 \\
			{\bf{Question:}} #3
		\end{minipage}
	\end{mdframed}
	\vspace{2mm}
}

\newcommand{\gs}{\textsc{StoryLineExtension}\xspace}

\newcommand{\gss}{\textsc{SLE}\xspace}

\newcommand{\ls}{\textsc{LocalStoryLineExtension}\xspace}

\newcommand{\lss}{\textsc{LSLE}\xspace}

\graphicspath{{figures/}}

\nolinenumbers

\begin{document}

\maketitle

\begin{abstract}
Storyline layouts visualize temporal interactions by drawing each character as an x-monotone curve and enforcing that the participants of every meeting form a contiguous vertical group. We study a drawing extension variant in which a layout of a sub-storyline is fixed and has to be extended by inserting $k$ missing characters while adhering to all meeting constraints. 
We consider two optimization objectives: minimizing the number of additional crossings introduced to complete the storyline, and minimizing the maximum number of crossings incurred by any single character. For both variants, we analyze the parameterized complexity with respect to the natural   parameter $k$ representing the size of the missing information, as well as additional structural parameters such as the number of characters $\sigma$ per time instant  and the number of  meetings $\mu$ involving missing characters. We contribute a broad collection of results, most of them tight, drawing a nearly complete picture of the complexity landscape of these extension problems.

\subparagraph{Generative AI Declaration}
Generative AI was not used in the preparation of this article.

\end{abstract}

\section{Introduction}

Storyline layouts (or \emph{storylines} for short) 
are visualizations of temporal interaction patterns of a set of actors. 
They have been initially proposed to show the plots of novels and movies~\cite{tm-dcosv-12}, but are more widely applicable for temporal co-authorship networks~\cite{gdlmt-svwua-20,hw-swp-24}, collaboration in software development~\cite{om-ses-10}, or rolling stock scheduling~\cite{dhnw-ows-25}. 
A storyline depicts time on the x-axis, shows a set of characters as x-monotone curves, and interactions (or \emph{meetings}) as vertical character groupings at the corresponding point in time (see \cref{fig:example}). 
Combinatorially, the primary degree of freedom in drawing a storyline is to determine valid vertical character orderings, where the characters in each meeting must form a contiguous interval, at all of the discrete meeting times. 
The predominant optimization goal in storylines is crossing minimization~\cite{knpss-mcsv-15,gjlm-cmsv-16,djjmmn-rmecmsd-24}; its \NP-hardness, even for only $\tau=2$ time steps, follows from that of bipartite crossing number~\cite{gj-cnn-83}; yet, other aspects, such as wiggle and whitespace minimization, have also been considered~\cite{dhnw-ows-25,trlcyw-iechs-18,thm-efgsvfsd-15}. 
While most existing approaches target the minimization of the total number of crossings, we are also interested in minimizing the maximum number of (local) crossings per character. 
This is motivated by the aim to balance the required crossings (and consequently the visual quality) among the different characters~\cite{plmnd-fmfmoswff-26a} and is related to the established concepts of $k$-planarity and $k$-plane graph drawings, in which each edge may have at most $k$ crossings~\cite{ht-bpg-20,dlm-sgdbp-18}. 

In this paper, we study both global and local crossing minimization in storylines as a \emph{drawing extension problem}~\cite{adfjkp-tppeg-15,ars-epod-21,bgkmn-eopgdft-23,dfgn-pcesl-24,ganian_et_al:LIPIcs.ICALP.2021.72,eghkn-ep1d-20}, i.e., we are given a (partial) storyline layout as well as a set of missing characters. The task is to optimally insert the missing characters into the partial layout with respect to (1)~the total number of resulting crossings or (2)~the local crossing number of both the old and the new characters. We call these problems \gs (\gss) and \ls (\lss), respectively. We take a parameterized complexity perspective and present several results, most of them tight, providing an almost complete picture of the complexity of \gss and \lss.

\begin{figure}[tb]
    \centering
    \includegraphics[]{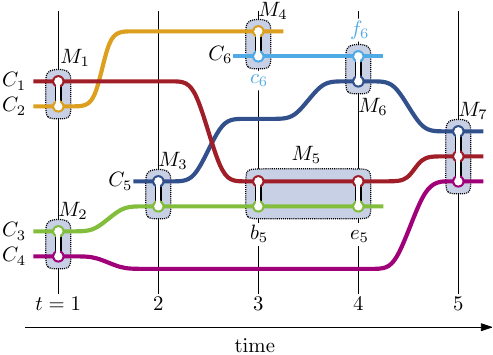}
    \caption{Sketch of a storyline layout with six characters, seven meetings, and five time instants.}
    \label{fig:example}
\end{figure}

\subparagraph{Preliminaries.}
A \emph{storyline} instance $\S = (\mathcal{C},\mathcal{M})$ consists of a set $\mathcal{C}=\{C_1, C_2, \dots, C_n\}$ of $n$ \emph{characters} and a set $\mathcal{M} = \{M_1, M_2, \dots, M_m\}$ of $m$ \emph{meetings}. 
Each meeting $M_i \in \mathcal{M}$ is a triple $\langle \mathcal{C}(M_i), b_i, e_i\rangle$, where $\mathcal{C}(M_i) \subseteq \mathcal{C}$ is a subset of characters, $b_i \in \mathbb{N}$ is the \emph{begin-time} of $M_i$ and $e_i \in \mathbb{N}$ is the \emph{end-time} of $M_i$ (thus, $b_i \leq e_i$). 
We say that a meeting $M_i$ is \emph{active} at any time $t$ with $b_i \le t \le e_i$, and that each character $C_j \in \mathcal{C}(M_i)$ \emph{participates} in $M_i$. 
A common assumption is that a character cannot participate in two distinct meetings at the same time instant. 
For each character $C_j \in \mathcal C$ we define its \emph{lifespan} $L_j=[c_j,f_j]$ as the interval between the begin-time $b_i$ of the first meeting $M_i$ that $C_j$ participates in and the end-time $e_l$ of the last meeting $M_l$ it participates in. 
For any time instant in its lifespan, we say a character is \emph{active}.
Let $A_t$ be the set of active characters at time $t$ and let $\sigma_t = |A_t|$ be their number.
By convention, we set to $1$ the minimum $c_j$ over all characters $C_j$, and we denote by $\tau$ the  largest $f_j$ over all characters $C_j$ (i.e.,  $\S$ is defined over $\tau$ time instants).

In a \emph{storyline layout} $\Gamma$ (see \cref{fig:example}), each character $C_j$ is represented as an x-monotone curve %
from time instant $c_j$ to time instant $f_j$. 
For each time instant $t \in [\tau] = \{1, 2, \dots, \tau\}$ the active character curves in $A_t$ are ordered vertically as a permutation $\pi_t$.
Two basic combinatorial properties must be satisfied: %
\begin{enumerate*}[label=(\roman*)]
\item For each meeting $M_i$, the curves representing the characters in $\mathcal{C}(M_i)$ are \emph{consecutive} in the permutations $\pi_j$ for all time instants $j \in [b_i,e_i]$, i.e., they form a contiguous group between $b_i$ and $e_i$; 
\item for meetings $M_i$ with a duration $e_i - b_i >0$ the characters $\mathcal C(M_i)$ cannot be crossed while $M_i$ is active.
\end{enumerate*}
For any two consecutive time instants $i$ and $i+1$, two character curves \emph{cross} if and only if they have a different relative order in $\pi_i$ and $\pi_{i+1}$. %
The \emph{local crossing number} of a storyline layout is the maximum number of crossings along any single character curve.

We are interested in extension problems. Let  $\S = (\mathcal C, \mathcal{M})$ be a storyline instance. A \emph{sub-storyline of $\S$} is a storyline instance $\S' = (\mathcal C' \subseteq \mathcal C,  \mathcal{M'})$ such that if a meeting in $\S$  contains at least one character of $\S'$, then such a meeting exists also in $\S'$ (even if incomplete).  More formally, for each meeting $M_i = \langle \mathcal{C}(M_i) , b_i, e_i \rangle\in \mathcal{M}$  such that $\mathcal{C}(M_i) \cap \mathcal{C'} \neq \emptyset$,  there is a meeting $M'_i \in \mathcal{M}'$ such that $M'_i=\langle \mathcal{C}(M_i) \cap \mathcal{C'}, b_i, e_i \rangle$. For a storyline layout $\Gamma$ of $\S$ we denote the induced \emph{sub-storyline layout} $\Gamma[\S']$ as the layout $\Gamma$ restricted to the characters and meetings in $\S'$. We focus on the following two decision problems:

\defquestion{(\textsc{Local}) \gs (\textsc{(L)}\gss)}{A  storyline instance $S = (C, \mathcal{M})$, an integer $\chi$, a storyline layout $\Gamma'$ of a sub-storyline $S' = (C' \subseteq C,  \mathcal{M'})$.}{Does there exist a storyline layout $\Gamma$ of $\S$ such that $\Gamma[\S'] = \Gamma'$ with (local) crossing number at most $\chi$?}

We define the following parameters relevant to quantify the complexity of (L)\gss.
\begin{itemize}
    \item $n = |\mathcal C'|$, number of characters in the sub-storyline $\S'$,
    \item $k = |\mathcal C \setminus \mathcal C'|$, number of new characters to be inserted,
    \item $\tau$, number of time instants of $\S$,
    \item $\mu$, number of meetings that involve new characters in $\mathcal C \setminus \mathcal C'$,
    \item $\sigma$, maximum number of active characters at any time instant.
\end{itemize}

\subparagraph{Contributions.}
In \cref{sec:hardness} we prove that \lss is \W[1]-hard, parameterized by the number $k$ of new characters and the maximum number $\sigma$ of characters per time instant, even if the missing characters participate in only $\mu=2$ meetings. This immediately implies para\NP-hardness of \lss for the parameter $\mu$.
Complementary to this hardness, we provide, in \cref{sec:XP}, parameterized algorithms for the parameters $\sigma$ and the (local or global) number $\chi$ of crossings of the extended storyline layout, based on dynamic programming (DP). More precisely, we show that \lss is in \XP\ parameterized by $k$ or $\sigma$ and in \FPT\ parameterized by $\sigma + \chi$. For \gss, these algorithms can be adjusted to run in \FPT-time for the parameter $\sigma$ and in \XP-time for $k$. Further, we can adapt that DP-algorithm to show that \gss is in \FPT~for $\mu + \chi$ and, hence, in \XP\ for $\mu$ -- in contrast to \lss. Lastly, in \cref{sec:onec}, we provide a polynomial-time algorithm for \gss if we are inserting a single new character.

\section{Hardness of \ls}\label{sec:hardness}

\begin{restatable}[th:lse-k]{theorem}{thLSEk}
    \label{th:lse-k}
    \lss parameterized by $k+\sigma$ is \W[1]-hard, even for $\mu=2$.
\end{restatable}

In order to prove \cref{th:lse-k}, we exploit a reduction from the \textsc{Unary Bin Packing} problem. Here, we are given a set $I=\{x_1,x_2,\dots,x_n\}$ of positive integers encoded in unary, and two integers $B$ and $K$. The task is to
decide whether $I$ can be partitioned into $K$ disjoint subsets (bins) such that the integers in each bin sum up to at most $B$. This problem has been proven to be W[1]-hard parameterized by the number $K$ of bins~\cite{DBLP:journals/jcss/JansenKMS13}. It is not difficult to see that the same W[1]-hardness holds for \textsc{Exact Unary Bin Packing} (\textsc{EUBP} for short), where the sum of the integers in each bin must be exactly $B$ (see, e.g.,~\cite{DBLP:journals/algorithmica/BlazejJRS25}). A necessary condition for an instance of \textsc{EUBP} to admit a solution is that $\sum_{j=0}^n x_j=K \cdot B$.

\subparagraph{Overview.} The idea is to transform an instance $\langle I, K, B \rangle$  of \textsc{EUBP} into an instance $\langle \S, \chi, \Gamma', \S' \rangle$ of \lss as follows. The full storyline instance $\S$ and its sub-storyline $\S'$ differ by exactly $k=K$ characters. In order to extend the layout $\Gamma'$ of $\S'$, the $k$ new characters are forced to go through specific regions in $\Gamma'$ in which they collect a controlled number of crossings. In particular, some special regions will force the new characters to collect a number of crossings that corresponds to an integer of $I$. Crucially, the structure of $\Gamma'$ is such that, if a valid extension $\Gamma$ exists, then the number of crossings collected by a single character over all special regions is exactly $B$. Moreover, no two characters can go through the same regions, which implies a partition of the integers (encoded by crossings) into bins (encoded by characters) such that the integers in each bin sum up to precisely $B$.
We now proceed by describing the gadgets  used to construct the instance $\langle \S, \chi, \Gamma', \S' \rangle$  of \lss. These gadgets should be intended as parts of $\Gamma'$ (and $\S'$) that will be eventually assembled together. 

\begin{figure}[h]
    \centering
    \includegraphics[page=2]{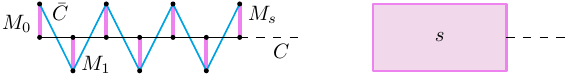}
    \caption{An $s$-saturator (left) and its schematic representation (right).}
    \label{fig:saturator}
\end{figure}

\subparagraph{Saturator gadgets.} An \emph{$s$-saturator gadget}, for any integer $s$, consists of two characters $C$ and $\bar{C}$, and $s+1$ meetings $M_0,M_1,\dots,M_{s}$, such that for each $i \in [s]$: $\mathcal{C}(M_i)=\{C,\bar{C}\}$, and $b_i=e_i=b^*+i$, where $b^*$ denotes the initial time instant when the gadget appears in the storyline. Also, in the corresponding layout in $\Gamma'$, shown in \cref{fig:saturator}, $C$ and $\bar{C}$ cross $s$ times. One of the two characters, say $C$, is called the \emph{leading} character of the gadget, and it will participate in additional meetings beyond the gadget itself, receiving at most $\chi-s$ additional crossings. 
Instead, the curve representing $\bar{C}$ starts and ends within the gadget.

\begin{figure}[h]
    \centering
    \includegraphics{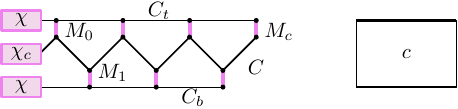}
    \caption{A $c$-channel (left) and its schematic representation (right).}
    \label{fig:channel}
\end{figure}

\subparagraph{Channel gadgets.} A \emph{$c$-channel gadget}, for a given integer $c$ called \emph{capacity}, consists of three characters $C$, $C_t$, $C_b$, and $c+1$ meetings $M_0,M_1,\dots,M_{c}$, such that for each $i \in [c]$: $\mathcal{C}(M_i)=\{C,C_t\}$ if $i$ is even or $\mathcal{C}(M_i)=\{C,C_b\}$ if $i$ is odd, and $b_i=e_i=b^*+i$, where $b^*$ denotes the initial time instant when the gadget appears in the storyline.  We call $C$ the \emph{central character} of the gadget, while $C_t$ and $C_b$ are \emph{boundary characters}. In addition, $C_t$ and $C_b$ are the leading characters of one $\chi$-saturator gadget each, while $C$ is the leading character of a $\chi_c$-saturator gadget, with $\chi_c = \chi-c$; see \cref{fig:channel}.
Intuitively, the gadget represents a region whose traversal implies crossing $c$ times the central character.%

\begin{figure}[h]
    \centering
    \includegraphics{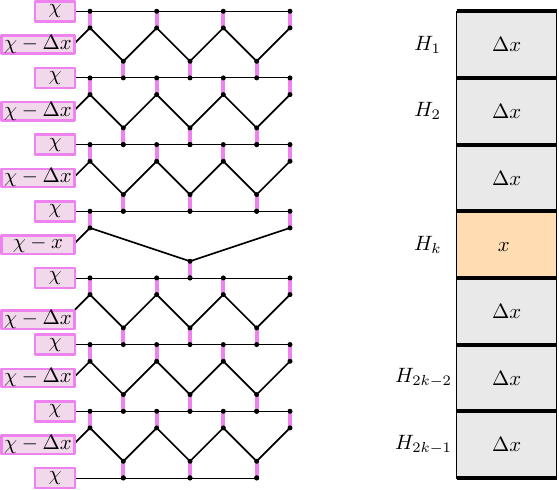}
    \caption{An $x$-column gadget for $x=2$ and $\Delta=3$ (left) and its schematic representation (right). The sparse channel has a light orange background, the dense channels are grey.}
    \label{fig:column}
\end{figure}

\subparagraph{Column gadgets.} An \emph{$x$-column gadget}, for an integer $x$,  consists of $2k-1$ channel gadgets $H_1,H_2,H_{2k-1}$ such that one boundary character of $H_i$ coincides with one boundary character of $H_{i+1}$, for $i \in [2k-2]$. Moreover, for each $H_i$, the capacity is equal to $\Delta \cdot x$ if $i \neq k$, and to $x$ if $i=k$. (We need $\Delta>1$, the exact value will be defined later.) Indeed, we call the channel $H_k$ the \emph{sparse channel} of the gadget, while all other channels are  \emph{dense}.  The corresponding layout is shown in \cref{fig:column}. %
Intuitively, each of the $k$ new characters will be forced to traverse a distinct channel in the column gadget.

\begin{figure}[h]
    \centering
    \includegraphics[scale=0.9]{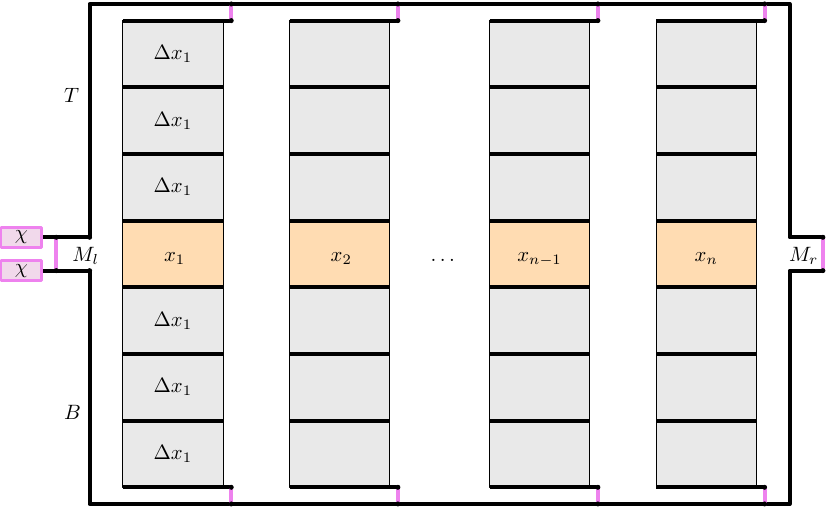}
    \caption{A schematic representation of $\Gamma'$.}
    \label{fig:all}
\end{figure}

\subparagraph{Putting all together.} Consider now $n$ column gadgets $X_1, X_2, \dots, X_n$, such that the capacity of the sparse channel of $X_i$ is equal to the integer $x_i \in I$, for each $i \in [n]$. Let $T$ and $B$ be two further characters that participate in two meetings $M_l$ and $M_r$, such that: $M_l$ occurs at the first time instant of $\S'$ (except for two saturator gadgets defined later), while $M_r$ occurs at the last time instant; $T$ and $B$ are the only two characters involved in these two meetings. 
Moreover, $T$ meets with one boundary character of the channel $H_1$ of the column $X_i$, for each $i \in [n]$. Analogously, $B$ meets with one boundary character of the channel $H_{2k-1}$ of the column $X_i$, for each $i \in [n]$. Also,  $T$ and $B$ are the leading characters of two $\chi$-saturators. This concludes the construction of $\Gamma'$ and hence of $\S'$; see also \cref{fig:all}. 

The following two properties easily follow from the described construction.

\begin{property}\label{pr:central}
   Every character of~$\S'$ is crossed $\chi$ times in $\Gamma'$, except for the central characters in the the channel gadgets, and for the non-leading characters in the saturator gadgets. 
\end{property}

Based on \cref{pr:central}, we say that a character $N$ \emph{traverses} a $c$-channel gadget $H$ in a layout, if the curve representing $N$ traverses $H$ monotonically from left to right, such that the central character of $H$ is crossed exactly $c$ times by $N$. 

\begin{property}\label{pr:channel}
   In any solution of \lss, a channel can be traversed by at most one character.
\end{property}

\subparagraph{Construction of $\S$.} It remains to construct $\S$ from $\S'$, which implies introducing $k$ new characters, which we denote by $N_1, N_2, \dots, N_k$, and the corresponding meetings, see \cref{fig:ex}. In particular, we add two new meetings. One meeting at the time instant right after the one of $M_l$, where all these new characters meet together with $T$; the second meeting at the time instant just before $M_r$, where all these new characters meet together with $B$.
Finally, we fix the values of $\Delta$ and $\chi$ as $\Delta = 2$ and  $\chi = B+\Delta\cdot B \cdot (K-1)$.

\smallskip This concludes the reduction; we note that the constructed storyline instance $\S$ is such that: (1) $k=K$ and $\mu=2$; (2) $\sigma_i \le 9k$, for each $i \in [\tau]$; (3) it has size $O(n \cdot k)$. An example of a solution of \lss constructed from a solution of \textsc{EUBP} can be seen in \cref{fig:ex}. The full proof of \cref{th:lse-k} is in \cref{apx:hardness}.

The constructed storyline instance $\S$ is such that: (1) $k=K$ and $\mu=2$; (2) $\sigma_i \le 9k$, for each $i \in [\tau]$; (3) it has size $O(n \cdot k)$.

\begin{figure}
    \centering
    \includegraphics{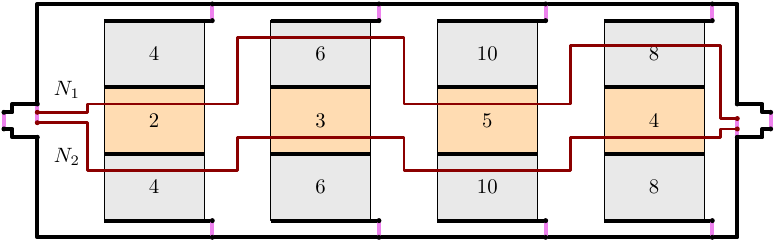}
    \caption{A solution $\Gamma$ of \lss constructed from a solution $\ell$ of \textsc{EUBP}. The instance of \textsc{EUBP} is $I=\{2,3,5,4\}$, $K=2$, and $B=7$. The main parameters of  $\Gamma$ are $k=2$, $\chi = 7+2\cdot7=21$.}
    \label{fig:ex}
\end{figure}

\section{Parameterized algorithms for \gss and \lss}\label{sec:XP}

In this section, we present algorithms for \gss and \lss parameterized by $k$, $\sigma$, $\chi$, and $\mu$.

\subsection{Parameterization in $k$, $\sigma$, and $\chi$}

We present an algorithm for \ls that shows membership in \XP\  parameterized by $k$ or $\sigma$ and in \FPT\  parameterized by $\sigma+\chi$. It can be adjusted to solve \gs in \FPT\ time  parameterized by $\sigma$ and in \XP\ time  parameterized by $k$.

Let $\langle \S, \chi, \Gamma', \S'\rangle$ be an instance of \lss. We use dynamic programming.

\begin{figure}
    \centering
    \includegraphics[page=1]{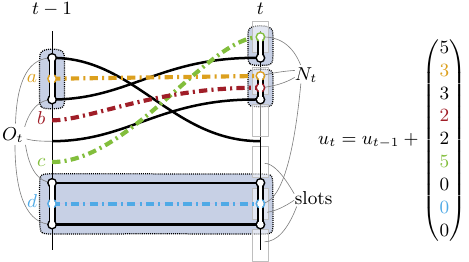}
    \caption{Illustration for the dynamic program. New characters are drawn dash-dotted.} %
    \label{fig:dp}
\end{figure}

\subparagraph{The DP states.}
Consider the discrete time instants $t=1,\ldots,\tau$. %
Let $O_t=A_t \cap \mathcal C'$ be the set of \emph{existing} active characters at time $t$ and let $N_t=A_t\setminus O_t$ be the set of \emph{new} active characters at $t$.
In $\Gamma'$, the characters in $O_t$ appear in some fixed vertical order $\pi_t$, 
and partition the vertical line at time $t$ into $|O_t|+1\le \sigma_t+1\le\sigma+1$ \emph{slots}; see \cref{fig:dp}.

A \emph{placement} of the new characters at time $t$ must specify which slot each new character occupies, as well as the relative order of the new characters within the same slot. 
Accordingly, we define a placement $P_t$ as a pair $(s_t,\prec_t)$, where $s_t \colon N_t\to\{0,1,\ldots,|O_t|\}$ assigns each new character to a slot, and $\prec_t$ is a permutation of $N_t$ that aligns with the slot assignment, i.e., every character in slot $i$ comes before every character in slot $j>i$. From $(s_t,\prec_t)$, we can obtain a single total order $\rho_t$ of $A_t$ for all characters at time $t$.

We process the time instants from left to right. We define a DP state at time $t$ as a pair $(P_t, u_t)$, where
\begin{enumerate*}[label=(\roman*)]
    \item $P_t$ is a placement of the new characters at time $t$, and
    \item $u_t$ is a vector that assigns to every character $C\in O_t$ the number $u_t(C)$ of crossings involving $C$ that occur strictly before time $t$.
\end{enumerate*}
The state $(P_t,u_t)$ is \emph{feasible} if there exists a partial extension of $\Gamma'$ up to time $t$ that realizes the placement $P_t$ at time $t$, respects all meeting contiguity constraints for the meetings up to time $t$, and every character $C\in A_t$ has $u_t(C)\le\chi$ crossings.

\subparagraph{Filling the table.}
At time $t=1$, we enumerate all placements $P_1$ of $N_1$ that are consistent with the meeting constraints, and set $u_1$ to the all-zero vector. At the final time $t=\tau$, we accept the instance if there exists any feasible state $(P_{\tau},u_{\tau})$.%

We now describe how to move from time $t-1$ to time $t$. Consider a feasible state $(P_{t-1},u_{t-1})$. We enumerate all possible placements $P_t$ of the active new characters $N_t$ that 
\begin{enumerate*}[label=(\roman*)]
    \item satisfy all meeting constraints at time $t$, and
    \item do not have any character cross an ongoing meeting between time $t-1$ and $t$ that it does not participate in.
\end{enumerate*}

For every character $C\in A_{t-1}\cap A_t$, we can compute the number $cr_t(C)$ of crossings involving $C$ that occur in the open strip $(t-1,t)$ from $(P_{t-1},P_t)$: this is exactly the number of characters whose relative vertical order with $C$ differs between the two total orders. Notice that this crossing count can optionally consider crossings in a weighted fashion; this will be needed in the proof of Theorem~\ref{th:se-mu}. 

We define the updated budget vector $u_t(C)$ for every character $C\in A_t$ as follows. If $C\in A_{t-1}$, then $u_t(C)=u_{t-1}(C)+cr_t(C)$; otherwise, the character $C$ starts at time $t$, and we set $u_t(C)=0$. If $u_t(C)\le\chi$ for every $C\in A_t$, then we store $(P_t,u_t)$ as a feasible state. All such transitions generate the set of feasible DP states at time $t$.

\begin{restatable}[th:lse-s]{theorem}{thLseS}
    \label{th:lse-s}
    \lss can be solved in $\tau\cdot (\sigma+\chi)^{O({\min\{k,\sigma\})}}$ time.
\end{restatable}
\begin{proof}[Sketch]
The number of possible placements at time $t$ is at most $(\sigma+1)^{|N_t|}\cdot |N_t|!\in\sigma^{O(\min\{k,\sigma\})}$, and for each placement the budget vector contributes a factor of $(\chi+1)^{O(\min\{k,\sigma\})}$, giving $(\sigma+\chi)^{O(\min\{k,\sigma\})}$ states per layer. 
A state $(P_{t-1},u_{t-1})$ transitions to $(P_t,u_t)$ if $P_t$ satisfies all meeting-contiguity constraints at $t$, no character crosses an ongoing meeting it does not participate in, and the updated budgets $u_t(C)=u_{t-1}(C)+cr_t(C)$ remain within $\chi$; both checks take $O(\sigma^2)$ time. 
A straightforward induction shows that a state is marked feasible if and only if a valid partial extension realizing it exists.
Accepting whenever any feasible state exists at $t=\tau$ and multiplying the per-layer cost by $\tau$ yields the claimed bound.
\end{proof}

Note that the above algorithm can be easily adjusted to work for \gss. Namely, for each slot assignment, we only need to store the minimum number of crossings it requires. 

\begin{corollary}
    \gss can be solved in $\tau\cdot \sigma^{O(\min\{k,\sigma\})}$ time.
\end{corollary}

\subsection{Parameterization in $\mu$}

Note that, by \cref{th:lse-k}, \ls is para-\NP-hard parameterized by $\mu$. In contrast, we can prove that \gs is \XP~under this parameter. In particular, we show that the problem is \FPT~in $\mu+\chi$.

\begin{restatable}[th:se-mu]{theorem}{lemSeMu}
    \label{th:se-mu}
    \gss can be solved in $\tau\cdot (\mu+\chi)^{O(\mu)}$ time.
\end{restatable}
\begin{proof}[Sketch] We partition the set of new characters into $2^\mu$ equivalent classes based on the meeting they participate to. The intuition is that characters in the same class can be drawn as a single bundle, as they participate in the same meetings. Next, we replace the characters in the same class with a single representative and apply the DP algorithm described in the previous section, with the only difference that crossings among representative characters are weighted based on the cardinality of the corresponding classes.
\end{proof}

\section{A polynomial-time algorithm for \gss with one new character}\label{sec:onec}

\begin{figure}
    \centering
    \includegraphics[page=2]{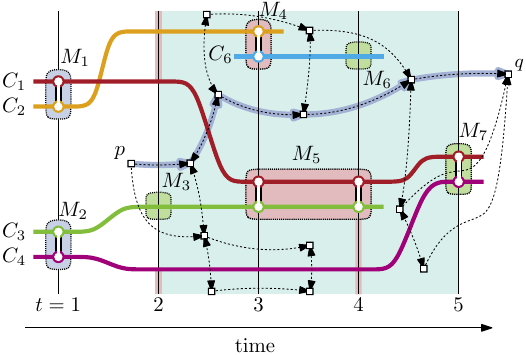}
    \caption{Illustration for the proof of \cref{thm:k1-poly}. Adding a character with lifespan $[2,5]$ that participates in the green meetings $M_3$, $M_6$, and $M_7$. The marked path through the dual corresponds to a solution where the character has one crossings, shown in \cref{fig:example}.}
    \label{fig:dual}
\end{figure}

We now consider the case $k = 1$, i.e., a single new character is inserted into
an existing storyline layout.  We show that this case admits a polynomial-time
algorithm by reducing the insertion problem to a shortest-path computation in a
planar graph derived from $\Gamma'$.

\begin{restatable}[thm:k1-poly]{theorem}{thmKPoly}
\label{thm:k1-poly}
  If $k = 1$, then \gss can be solved in $O(\tau \cdot \sigma + \chi)$ time.
\end{restatable}
\begin{proof}[Sketch]
Let $C_\mathrm{new}$ be the single new character with lifespan $[c, f]$.
We build a planar directed graph~$G$ whose nodes are the cells of the arrangement induced by $\Gamma'$ over $[c,f]$; see \cref{fig:dual}. 
Edges of $G$ cross character curves (in both directions) or vertical lines through time instants (left to right only), with a source $p$ at $t=c$ and a sink $q$ at $t=f$. 
We enforce meeting constraints by pruning: only cells inside or adjacent to a meeting block are retained if $C_\mathrm{new}$ is a participant, and
edges crossing the boundary of the block are removed otherwise. 
Every path from $p$ to $q$ crosses each vertical line exactly once, so a path of length $\ell$ corresponds to inserting $C_\mathrm{new}$ with $\ell - (f-c+1)$ crossings; a BFS on $G$ thus finds the optimal insertion.
Since each crossing in $\Gamma'$ splits one cell into two, $G$ has $O(\tau\cdot\sigma+\chi)$ nodes and edges, yielding the claimed running time.
\end{proof}

\section{Conclusions}

Our research draws a nearly complete picture of the complexity landscape of \gss and \lss. The key open question is whether \gss parameterized in $k$ lies in \FPT. Towards this question, one could consider the combination $k+\chi$ as a parameter.

\bibliographystyle{abbrvurl}
\bibliography{bibliography}

\begin{thebibliography}{10}

\bibitem{adfjkp-tppeg-15}
P.~Angelini, G.~Di~Battista, F.~Frati, V.~Jelínek, J.~Kratochvíl, M.~Patrignani, and I.~Rutter.
\newblock Testing planarity of partially embedded graphs.
\newblock {\em ACM Trans. Algorithms}, 11(4):32:1--32:42, 2015.
\newblock \href {https://doi.org/10.1145/2629341} {\path{doi:10.1145/2629341}}.

\bibitem{ars-epod-21}
P.~Angelini, I.~Rutter, and T.~P. Sandhya.
\newblock Extending partial orthogonal drawings.
\newblock {\em J. Graph Algorithms Appl.}, 25(1):581--602, 2021.
\newblock \href {https://doi.org/10.7155/jgaa.00573} {\path{doi:10.7155/jgaa.00573}}.

\bibitem{bgkmn-eopgdft-23}
S.~Bhore, R.~Ganian, L.~Khazaliya, F.~Montecchiani, and M.~Nöllenburg.
\newblock Extending orthogonal planar graph drawings is fixed-parameter tractable.
\newblock In E.~W. Chambers and J.~Gudmundsson, editors, {\em Computational Geometry (SoCG'23)}, volume 258 of {\em LIPIcs}, pages 18:1--18:16. Schloss Dagstuhl -- Leibniz-Zentrum für Informatik, 2023.
\newblock \href {https://doi.org/10.4230/LIPIcs.SoCG.2023.18} {\path{doi:10.4230/LIPIcs.SoCG.2023.18}}.

\bibitem{DBLP:journals/algorithmica/BlazejJRS25}
V.~Blazej, S.~Jana, M.~S. Ramanujan, and P.~Strulo.
\newblock On the parameterized complexity of {E}ulerian strong component arc deletion.
\newblock {\em Algorithmica}, 87(11):1669--1709, 2025.
\newblock \href {https://doi.org/10.1007/S00453-025-01336-6} {\path{doi:10.1007/S00453-025-01336-6}}.

\bibitem{dfgn-pcesl-24}
T.~Depian, S.~D. Fink, R.~Ganian, and M.~Nöllenburg.
\newblock The parameterized complexity of extending stack layouts.
\newblock In S.~Felsner and K.~Klein, editors, {\em Graph Drawing and Network Visualization (GD'24)}, volume 320 of {\em LIPIcs}, pages 12:1--12:17. Schloss Dagstuhl -- Leibniz-Zentrum für Informatik, 2024.
\newblock \href {https://doi.org/10.4230/LIPIcs.GD.2024.12} {\path{doi:10.4230/LIPIcs.GD.2024.12}}.

\bibitem{gdlmt-svwua-20}
E.~{Di Giacomo}, W.~Didimo, G.~Liotta, F.~Montecchiani, and A.~Tappini.
\newblock Storyline visualizations with ubiquitous actors.
\newblock In D.~Auber and P.~Valtr, editors, {\em Graph Drawing and Network Visualization (GD'20)}, volume 12590 of {\em LNCS}, pages 324--332. Springer, 2020.
\newblock \href {https://doi.org/10.1007/978-3-030-68766-3_25} {\path{doi:10.1007/978-3-030-68766-3_25}}.

\bibitem{dlm-sgdbp-18}
W.~Didimo, G.~Liotta, and F.~Montecchiani.
\newblock A survey on graph drawing beyond planarity.
\newblock {\em ACM Computing Surveys}, 52(1):4:1--4:37, 2019.
\newblock \href {https://doi.org/10.1145/3301281} {\path{doi:10.1145/3301281}}.

\bibitem{dhnw-ows-25}
A.~Dobler, T.~Hegemann, M.~Nöllenburg, and A.~Wolff.
\newblock Optimizing wiggle in storylines.
\newblock In V.~Dujmović and F.~Montecchiani, editors, {\em Graph Drawing and Network Visualization (GD'25)}, volume 357 of {\em LIPIcs}, pages 39:1--39:17. Schloss Dagstuhl -- Leibniz-Zentrum für Informatik, 2025.
\newblock \href {https://doi.org/10.4230/LIPIcs.GD.2025.39} {\path{doi:10.4230/LIPIcs.GD.2025.39}}.

\bibitem{djjmmn-rmecmsd-24}
A.~Dobler, M.~Jünger, P.~J. Jünger, J.~Meffert, P.~Mutzel, and M.~Nöllenburg.
\newblock Revisiting {ILP} models for exact crossing minimization in storyline drawings.
\newblock In S.~Felsner and K.~Klein, editors, {\em Graph Drawing and Network Visualization (GD'24)}, volume 320 of {\em LIPIcs}, pages 31:1--31:19. Schloss Dagstuhl -- Leibniz-Zentrum für Informatik, 2024.
\newblock \href {https://doi.org/10.4230/LIPIcs.GD.2024.31} {\path{doi:10.4230/LIPIcs.GD.2024.31}}.

\bibitem{eghkn-ep1d-20}
E.~Eiben, R.~Ganian, T.~Hamm, F.~Klute, and M.~Nöllenburg.
\newblock Extending partial 1-planar drawings.
\newblock In A.~Czumaj, A.~Dawar, and E.~Merelli, editors, {\em Automata, Languages, and Programming (ICALP'20)}, volume 168 of {\em LIPIcs}, pages 43:1--43:19. Schloss Dagstuhl--Leibniz-Zentrum für Informatik, 2020.
\newblock \href {https://doi.org/10.4230/LIPIcs.ICALP.2020.43} {\path{doi:10.4230/LIPIcs.ICALP.2020.43}}.

\bibitem{ganian_et_al:LIPIcs.ICALP.2021.72}
R.~Ganian, T.~Hamm, F.~Klute, I.~Parada, and B.~Vogtenhuber.
\newblock {Crossing-Optimal Extension of Simple Drawings}.
\newblock In N.~Bansal, E.~Merelli, and J.~Worrell, editors, {\em Automata, Languages, and Programming (ICALP'21)}, volume 198 of {\em LIPIcs}, pages 72:1--72:17. Schloss Dagstuhl -- Leibniz-Zentrum f{\"u}r Informatik, 2021.
\newblock \href {https://doi.org/10.4230/LIPIcs.ICALP.2021.72} {\path{doi:10.4230/LIPIcs.ICALP.2021.72}}.

\bibitem{gj-cnn-83}
M.~R. Garey and D.~S. Johnson.
\newblock Crossing number is {NP}-complete.
\newblock {\em SIAM J. Algebr. Discr. Methods}, 4(3):312--316, 1983.
\newblock \href {https://doi.org/10.1137/0604033} {\path{doi:10.1137/0604033}}.

\bibitem{gjlm-cmsv-16}
M.~Gronemann, M.~Jünger, F.~Liers, and F.~Mambelli.
\newblock Crossing minimization in storyline visualization.
\newblock In Y.~Hu and M.~Nöllenburg, editors, {\em Graph Drawing and Network Visualization (GD'16)}, volume 9801 of {\em LNCS}, pages 367--381. Springer, 2016.
\newblock \href {https://doi.org/10.1007/978-3-319-50106-2_29} {\path{doi:10.1007/978-3-319-50106-2_29}}.

\bibitem{hw-swp-24}
T.~Hegemann and A.~Wolff.
\newblock Storylines with a protagonist.
\newblock In S.~Felsner and K.~Klein, editors, {\em Graph Drawing and Network Visualization (GD'24)}, volume 320 of {\em LIPIcs}, pages 26:1--26:22. Schloss Dagstuhl - Leibniz-Zentrum f{\"{u}}r Informatik, 2024.
\newblock \href {https://doi.org/10.4230/LIPICS.GD.2024.26} {\path{doi:10.4230/LIPICS.GD.2024.26}}.

\bibitem{ht-bpg-20}
S.-H. Hong and T.~Tokuyama, editors.
\newblock {\em Beyond Planar Graphs}.
\newblock Springer, 2020.

\bibitem{DBLP:journals/jcss/JansenKMS13}
K.~Jansen, S.~Kratsch, D.~Marx, and I.~Schlotter.
\newblock Bin packing with fixed number of bins revisited.
\newblock {\em J. Comput. Syst. Sci.}, 79(1):39--49, 2013.
\newblock \href {https://doi.org/10.1016/J.JCSS.2012.04.004} {\path{doi:10.1016/J.JCSS.2012.04.004}}.

\bibitem{knpss-mcsv-15}
I.~Kostitsyna, M.~Nöllenburg, V.~Polishchuk, A.~Schulz, and D.~Strash.
\newblock On minimizing crossings in storyline visualizations.
\newblock In E.~{Di Giacomo} and A.~Lubiw, editors, {\em Graph Drawing and Network Visualization (GD'15)}, volume 9411 of {\em LNCS}, pages 192--198. Springer, 2015.
\newblock \href {https://doi.org/10.1007/978-3-319-27261-0_16} {\path{doi:10.1007/978-3-319-27261-0_16}}.

\bibitem{om-ses-10}
M.~Ogawa and K.~Ma.
\newblock Software evolution storylines.
\newblock In A.~C. Telea, C.~G{\"{o}}rg, and S.~P. Reiss, editors, {\em Software Visualization (SoftVis'10)}, pages 35--42. ACM, 2010.
\newblock \href {https://doi.org/10.1145/1879211.1879219} {\path{doi:10.1145/1879211.1879219}}.

\bibitem{plmnd-fmfmoswff-26a}
T.~Piselli, G.~Liotta, F.~Montecchiani, M.~Nöllenburg, and S.~{Di Bartolomeo}.
\newblock {$F^2$Stories}: A modular framework for multi-objective optimization of storylines with a focus on fairness.
\newblock {\em {IEEE} Trans. Vis. Comput. Graph.}, 32(1):747--757, 2026.
\newblock \href {https://doi.org/10.1109/TVCG.2025.3634228} {\path{doi:10.1109/TVCG.2025.3634228}}.

\bibitem{thm-efgsvfsd-15}
Y.~Tanahashi, C.-H. Hsueh, and K.-L. Ma.
\newblock An efficient framework for generating storyline visualizations from streaming data.
\newblock {\em {IEEE} Trans. Vis. Comput. Graph.}, 21(6):730--742, 2015.
\newblock \href {https://doi.org/10.1109/TVCG.2015.2392771} {\path{doi:10.1109/TVCG.2015.2392771}}.

\bibitem{tm-dcosv-12}
Y.~Tanahashi and K.-L. Ma.
\newblock Design considerations for optimizing storyline visualizations.
\newblock {\em {IEEE} Trans. Vis. Comput. Graph.}, 18(12):2679--2688, 2012.
\newblock \href {https://doi.org/10.1109/TVCG.2012.212} {\path{doi:10.1109/TVCG.2012.212}}.

\bibitem{trlcyw-iechs-18}
T.~Tang, S.~Rubab, J.~Lai, W.~Cui, L.~Yu, and Y.~Wu.
\newblock {iStoryline}: Effective convergence to hand-drawn storylines.
\newblock {\em {IEEE} Trans. Vis. Comput. Graph.}, 25(1):769--778, 2019.
\newblock \href {https://doi.org/10.1109/TVCG.2018.2864899} {\path{doi:10.1109/TVCG.2018.2864899}}.

\end{thebibliography}

\clearpage

\appendix

\section{Missing Proofs of \cref{sec:hardness}}\label{apx:hardness}
We conclude the proof of \cref{th:lse-k}. 

\thLSEk*
\label{th:lse-k*}

At high level, the gadgets are  arranged in such a way that each new character $N_i$ corresponds to one bin. Each new character must be routed through some sparse and some dense channels and the capacities are designed in such a way that a solution with local crossing number $\chi$ exists if and only if the collected sparse channels add up to the bin capacity $B$.

\begin{lemma}\label{le:proof-1}
    Suppose that the instance $\langle I, K, B \rangle$  of \textsc{EUBP} admits a solution $\ell$. Then, the corresponding instance $\langle \S, \chi, \Gamma', \S' \rangle$ of \lss admits a solution $\Gamma$.
\end{lemma}
\begin{proof}
    In order to construct $\Gamma$ from $\Gamma'$, we must draw the curves corresponding to the $k$ new characters. In particular, we associate $N_i$ to the $i$-th subset of the solution $\ell$ as follows. If the integer $x_j$ belongs to this subset, then $N_i$ is routed through the sparse channel of $X_j$, otherwise to one of the dense channels of $X_j$. This can be done such that no two new characters cross each other, since we have $k-1$ dense channels above and below the sparse channel of every column. See \cref{fig:ex} for an illustration.
    
    To establish the validity of this solution, we shall verify that each character is crossed at most $\chi$ times. Recall that no channel is traversed by more than one character due to \cref{pr:channel}. The only existing characters being crossed by the new characters are the central characters of the traversed channels. Consider one such central character: it is crossed by exactly one new character by construction, and hence in total it receives $\chi_c + c = \chi - c + c = \chi$ crossings, as desired. On the other hand, each new character receives:
    \begin{itemize}
        \item $B$ crossings over all the traversed sparse channels, because the integers of any subset of $\ell$ sum up to exactly $B$;
        \item  $\Delta(\sum_{j=0}^n x_j - B)$ crossings over all the traversed dense channels, as $\sum_{j=0}^n x_j - B$ is the sum of the integers that are not part of the subset. On the other hand, we noted that $\sum_{j=0}^n x_j=k \cdot B$, hence $\Delta(\sum_{j=0}^n x_j - B)=\Delta \cdot B \cdot (K-1)$.
    \end{itemize}
    
\noindent Overall, $B+\Delta \cdot B \cdot (K-1) = \chi$, as desired.
\end{proof}

\begin{lemma}\label{le:proof-2}
    Suppose that the instance $\langle \S, \chi, \Gamma', \S' \rangle$ of \lss admits a solution $\Gamma$. Then, the corresponding instance $\langle I, K, B \rangle$  of \textsc{EUBP} admits a solution $\ell$.
\end{lemma}

\begin{proof}
    Consider the solution $\Gamma$ and the routing of the new characters $N_1$, $N_2$, $\dots$, $N_k$. By \cref{pr:central}, the only characters with fewer than $\chi$ crossings in $\Gamma'$ are the central characters of the channel gadgets, and the non-leading characters of the saturator gadgets.  One easily observes that crossings of new characters with non-leading characters of the saturator gadgets can always be avoided and hence, by focusing on a crossing-minimal solution, we shall assume that $\Gamma$ does not contain such crossings. 
    
    Also, by \cref{pr:channel}, we know that any channel is traversed by at most one new character. For any new character $N$, let $B' = B-\delta$ be the sum of the capacities of the sparse channels traversed by $N$. Then, overall, $N$ is crossed $B'+\Delta \cdot (\sum_{j=0}^n x_j -B')$ times, which can be written as  $B-\delta+\Delta \cdot (\sum_{j=0}^n x_j -B+\delta) = B-\delta + \Delta \cdot \delta + \Delta(\sum_{j=0}^n x_j -B) = \delta(\Delta-1)+B+ \Delta(\sum_{j=0}^n x_j -B)=\delta(\Delta-1)+B+ \Delta \cdot B \cdot (K-1)=\delta(\Delta-1)+\chi$.

    We claim that $\delta=0$, that is, $B'=B$.  If this is the case, a solution $\ell$ for \textsc{EUBP} can be constructed by mapping the $i$-th subset of $I$ to the integers corresponding to the capacities of the sparse channels traversed by $N_i$.
    
    Suppose first, for a contradiction, that $\delta>0$, that is, $B'< B$. Then, since $\Delta>1$, $N$ is crossed more than $\chi$ times, which contradicts the fact that $\Gamma$ is a valid solution of \lss.

    Suppose now, still for a contradiction, that $\delta < 0$, that is, $B'> B$. Since $\sum_{j=0}^n x_j=K \cdot B$, if $N$ traverses a set of sparse channels whose capacities sum up to $B'>B$, then there is a different new character $N^*$  which traverses a set of sparse channels whose capacities sum up to $B^*<B$. Then we can apply the above argument to show that $N^*$ is crossed more than $\chi$ times, which again contradicts the fact that $\Gamma$ is a valid solution of \lss.
\end{proof}

The next lemma concludes the proof by discussing the size of the constructed instance and its parameters.

\begin{lemma}\label{le:sigma}
The constructed storyline instance $\S$ is such that: (1) $k=K$ and $\mu=2$; (2) $\sigma_i \le 9k$, for each $i \in [\tau]$; (3) it has size $O(n \cdot k)$.
\end{lemma}
\begin{proof}
Item (1) immediately follows by  construction of $\S$.
Concerning item (2), recall that $\sigma_i$ denotes the number of active characters at time instant $i$. One easily verifies that the largest number of characters co-occur inside column gadgets. In particular, we shall focus on those time instants where the different saturator gadgets overlap. For each column, we have $4k-1$ saturators, each of which contains two characters, plus $T$, $B$, and the $k$ new characters. Thus in total we have $\sigma_i \le 9k$, for any time instant $i \in [\tau]$. 
Finally, item (3) is due to the fact that  each column gadget contains $O(k)$ characters and meetings, and that we have $n$ such gadgets.
\end{proof}

\smallskip This concludes the proof of the theorem; an example of a solution of \lss constructed from a solution of \textsc{EUBP} can be seen in \cref{fig:ex}. 

\section{Missing Proofs of \cref{sec:XP}}

\thLseS*
\label{th:lse-s*}
\vspace{-\topsep}
\begin{proof}
We prove by induction on $t$ that a pair $(P_t,u_t)$ is marked feasible by the DP if and only if there exists a partial extension of $\Gamma'$ up to time $t$ that realizes placement $P_t$, budgets $u_t$ and respects all constraints. 
The base case $t=0$ holds by construction. For the inductive step, assume that the claim holds for time $t-1$. If the DP marks $(P_t,u_t)$ as feasible, then it comes from some feasible $(P_{t-1},u_{t-1}$) and a transition that respects the meeting constraints at $t$ and updates budgets exactly by the crossings induced in strip $(t-1,t)$. Concatenating the corresponding partial drawing up to $t-1$ with the strip $(t-1,t)$ yields a valid partial drawing up to $t$ with the stated placement and budgets.

Conversely, any valid partial drawing up to $t$ restricts to a valid partial drawing up to $t-1$ for some $(P_{t-1},u_{t-1})$. Comparing the two consecutive time orders determines exactly the increments $cr_t(C)$, hence the DP includes the corresponding transition and marks $(P_t,u_t)$ as feasible. This establishes the invariant, and for $t=\tau$ it yields the equivalence between acceptance and existence of an extension with local crossing number at most $\chi$.

It remains to bound the running time in terms of $\sigma$, $\chi$, $k$, and $\tau$. Fix a time instant $t$. 

The number of possible slot assignments at $t$ is at most $(|O_t|+1)^{|N_t|}\le(\sigma+1)^{\min\{k,\sigma\}}$. For each slot assignment, the order of the $|N_t|$ new characters can be chosen in at most $|N_t|!\le(\min\{k,\sigma\})!$ ways. Thus, the number of possible placements satisfies
\[ |\{P_t\}|\le(|O_t|+1)^{|N_t|}\cdot |N_t|!\le(\sigma+1)^{\min\{k,\sigma\}}\cdot(\min\{k,\sigma\})!\in\sigma^{O(\min\{k,\sigma\})}.\]

For the budget vector, each entry is an integer in $\{0,\ldots,\chi\}$. %
Hence, for a fixed placement, the number of possible budget vectors is at most $(\chi+1)^{|N_t|}\le(\chi+1)^{O(\min\{k,\sigma\})}$. Therefore, the total number of states at time $t$ is bounded by
\[\sigma^{O(\min\{k,\sigma\})}\cdot(\chi+1)^{O(\min\{k,\sigma\})}=(\sigma+\chi)^{O(\min\{k,\sigma\})}.\]

For each state at time $t-1$, we enumerate all placements at time $t$, i.e., at most $\sigma^{O(\min\{k,\sigma\})}$ candidates. For a state $(P_t,u_t)$, we can compute the number of crossings in the strip $(t-1,t)$ in $O(\sigma^2)$ time by counting. As there can be at most $\sigma$ ongoing meetings, we can check in $O(\sigma^2)$ time that the meeting constraints are satisfied at time $t$ and that no character crosses an ongoing meeting between time $t-1$ and $t$ that it does not participate in.
Thus, we need $(\sigma+\chi)^{O(\min\{k,\sigma\})}\cdot\sigma^{O(\min\{k,\sigma\})}\cdot O(\sigma^2)$ time to find all feasible DP states at time $t$. Consequently, the total running time over all $\tau$ layers is
\[\tau\cdot(\sigma+\chi)^{O(\min\{k,\sigma\})}\cdot\sigma^{O(\min\{k,\sigma\})}\cdot O(\sigma^2)=\tau\cdot(\sigma+\chi)^{\min\{k,\sigma\}}.\qedhere\]
\end{proof}

\lemSeMu*
\label{th:se-mu*}
\vspace{-\topsep}
\begin{proof}
    Consider any two new characters $C$ and $C'$. We say that $C$ and $C'$ are equivalent, if and only if they participate in the same meetings. This relation partitions the set of new characters into $2^\mu$ equivalence classes. 
    
    Consider an equivalence class $S$ of characters. If a solution of \gs exists, we claim that there is a (possibly different) solution for which the following two conditions hold: (i) no two characters in $S$ cross each other, and (ii) in any meeting participated by the characters of $S$, such characters form contiguous subgroups in the meeting. Notably, if (i) and (ii) hold, then between any two time instants $t$ and $t'$, it follows that any two characters of $S$ cross exactly the same characters, and  in particular characters  of two equivalent classes are such that either  they all pairwise cross or they are all pairwise disjoint. Consequently, we can replace each set of characters in the same class with a single representative character. Next, we can apply the DP algorithm described in the previous section with one simple difference: if a new character $C$ crosses an existing character then the number of crossings to be counted is equal to the cardinality of the equivalence class of $C$, while if $C$ crosses another new character $C'$ then the number of crossings to be counted is equal to the sum of cardinality of the equivalence classes of $C$ and $C'$. Also, since $\sigma$ can be replaced with $\mu$ in this setting, the statement follows.

    It remains to prove the claim. If (i) or (ii) are violated in a solution of \gs,   
    there are two characters $C$ and $C'$ in the same equivalence class such that they cross each other, or there is a character $C''$ from a different equivalence class in some meeting whose position is between $C$ and $C'$. We call such a pair \emph{inconsistent} in the following. If a pair of inconsistent characters $C$ and $C'$ exists, we let $C$ be the character that collects at most the same number of crossings as $C'$. Next, we modify the solution by redrawing $C'$ such that it is always the character below $C$ in every meeting. Notice that, by the previous assumption, the new solution contains at most the same number of crossings as the previous one, and contains one less pair of inconsistent characters. By repeating this procedure for every pair of inconsistent characters, we obtain a solution that satisfies the claim.
\end{proof}

\section{Missing proofs of \cref{sec:onec}}

\thmKPoly*
\label{thm:k1-poly*}
\vspace{-\topsep}
\begin{proof}
  Let $C_\mathrm{new}$ be the single new character with lifespan $[c, f]$.
  Since $\Gamma'$ is fixed, all crossings among the characters in $\mathcal{C}'$ are
  already determined, so minimizing the total crossing number of $\Gamma$ reduces to
  minimizing the crossings incurred by $C_\mathrm{new}$ with the existing characters.
  We achieve this via a shortest-path computation on the dual of the arrangement
  induced by $\Gamma'$.

  Restrict $\Gamma'$ to the time interval $[c, f]$ and draw a vertical line at every
  time instant $t \in \{c, \ldots, f\}$.  The character curves together with
  these vertical lines induce a planar arrangement whose bounded faces we call
  \emph{cells}. We define a directed graph $G$ by placing a node inside every cell
  and adding edges between nodes whose cells share a boundary segment: if two cells
  are separated by a character curve, we add edges in both directions, and if they are
  separated by a vertical line, we add a directed edge from left to right.  We add a
  source $p$ connected to every cell whose left boundary lies on the vertical line at
  $t = c$, and a sink $q$ to which every cell whose right boundary lies on the
  vertical line at $t = f$ connects.
 
  We then prune $G$ to enforce the meeting constraints.  If a meeting
  $M_i$ does not involve $C_\mathrm{new}$, its participants form a consecutive,
  uncrossable block throughout $[b_i, e_i]$, so $C_\mathrm{new}$ may not cross into
  or out of the block; we remove all edges that cross a boundary curve of
  $\mathcal{C}(M_i)$ for $t \in [b_i, e_i]$.  If $M_i$ does involve $C_\mathrm{new}$,
  the full participant set $\mathcal{C}(M_i) \cup \{C_\mathrm{new}\}$ must be
  consecutive; since $\mathcal{C}(M_i) \cap \mathcal{C}'$ already forms a consecutive
  block in $\Gamma'$, we remove all cells that are neither inside nor immediately
  adjacent to this block for $t \in [b_i, e_i]$, and at times $b_i$ and $e_i$ we
  additionally remove all vertical-line edges that are not inside or immediately adjacent to the block.
  The remaining case is a meeting consisting solely of $C_\mathrm{new}$ and one
  character $C_j \in \mathcal{C}'$; we retain only the two cells
  neighboring $C_j$ for $t \in [b_i, e_i]$ and remove all vertical-line edges at
  times $b_i$ and $e_i$ that do not involve those cells.
 
  Every path from $p$ to $q$ in $G$ crosses each of the $f - c + 1$ vertical lines
  exactly once. Hence, every path from $p$ to $q$ of length $x$ corresponds to inserting character $C_\mathrm{new}$ with $x - f + c - 1$ crossings with existing characters.  The
  answer to \gss is \emph{yes} if and only if this number does not exceed the
  remaining crossing budget.
 
  Since $G$ is planar, its number of edges is linear in its number of nodes.
  Without any crossings in $\Gamma'$, there are $O(\tau \cdot \sigma)$ cells; every
  crossing in $\Gamma'$ splits one cell into two, so the total number of cells, and
  hence nodes and edges of $G$, is $O(\tau \cdot \sigma + \mathrm{cr}(\Gamma'))\subseteq O(\tau \cdot \sigma + \chi)$. We can therefore use a BFS traversal that runs in $O(\tau \cdot \sigma + \chi)$ time to solve \gss.
\end{proof}
\end{document}